\documentclass{article}
\usepackage{amsmath,amssymb,amsthm} % assumes amsmath package installed
\usepackage{enumerate}
\usepackage{graphicx} % for pdf, bitmapped graphics files
\usepackage{xcolor}
\usepackage{subfigure}
\usepackage[T1]{fontenc}
\usepackage[utf8]{inputenc}
\usepackage{icomma}
\usepackage{authblk}

\definecolor{fernando}{rgb}{0, 0, 0.7}
\definecolor{taya}{rgb}{0.3, 0.7, 0.3}

\newcommand{\AC}{\mathcal{A}}
\newcommand{\HH}{\mathcal{H}}
\newcommand{\RR}{\mathbb{R}}
\newcommand{\ra}{\mathrm{a}}
\newcommand{\rd}{\mathrm{d}}
\newcommand{\rh}{\mathrm{h}}
\newcommand{\rn}{\mathrm{n}}
\newcommand{\rs}{\mathrm{s}}
\newcommand{\ru}{\mathrm{u}}
\newcommand{\opt}{\mathrm{opt}}
\newcommand{\fix}{\mathrm{fix}}

\DeclareMathOperator{\sgn}{Sgn}

\newtheorem{proposition}{Proposition}
\newtheorem{corollary}{Corollary}

\theoremstyle{definition}
\newtheorem{assumption}{Assumption}

\theoremstyle{remark}

\title{
Formalizing Neuromorphic Control Systems\\ \Large A General Proposal and A Rhythmic Case Study
}

\author[1]{Taisia Medvedeva}
\author[2]{Alessio Franci}
\author[1]{Fernando Castaños}
\affil[1]{Automatic Control Department, Cinvestav, Mexico}
\affil[2]{Department of Electrical Engineering and Computer Science at the University of Liège, Belgium}

\begin{document}

\maketitle

\begin{abstract}
    
    Neuromorphic control is receiving growing attention due to the multifaceted advantages it brings over more classical control approaches, including: sparse and on-demand sensing, information transmission, and actuation; energy efficient designs and realizations in neuromorphic hardware; event-based signal processing and control signal computation. However, a general control-theoretical formalization of what ``neuromorphic control systems'' are and how we can rigorously analyze, design, and control them is still largely missing. In this note, we suggest a possible path toward formalizing neuromorphic control systems. We apply the proposed framework to a rhythmic control case study and rigorously show how it has the potential to make neuromorphic control systems analysis and design amenable to mature control theoretical approaches like describing function analysis and harmonic balance, fast-slow analysis, discrete and hybrid systems, and robust optimization. 
    
\end{abstract}

%%%%%%%%%%%%%%%%%%%%%%%%%%%%%%%%%%%%%%%%%%%%%%%%%%%%%%%%%%%%%%%%%%%%%%%%%%%%%%%%
\section{INTRODUCTION}

The original aim of Neuromorphic Engineering (NE) was to design large-scale analog electronic circuits that mimic the architecture and functions of biological nervous systems~\cite{mead1990neuromorphic}. Across the years, NE diversified in a variety of different approaches and goals, like the design of low-power sensors~\cite{Lichtsteiner2008dvs} and neural processors~\cite{moradi2017scalable}, memristive materials~\cite{john2022reconfigurable}, and many blends and interactions with classical (digital) machine learning~\cite{yik2025neurobench}. But all recent and less recent neuromorphic approaches share a fundamental neuro-inspired property: the ``spiking'' or ``event-based'' nature of communication and computation. The asynchronous, on-demand nature of event-based signal processing and computing ensures speed, parsimonious power consumption, and scalability by overcoming the Von Neumann digital dichotomy between memory and logic units: like neuron spikes in a brain, each event is both a memory and a unit of computation.

Although control theory was close to neuromorphic engineering at its origin~\cite{de1990neuron}, the two fields have remained relatively apart until recently, with the appearance of both control-theoretical approaches to neuromorphic systems design~\cite{castanos2017implementing,ribar2021} and the use of neuromorphic event-based sensors in fast feedback control loops~\cite{singh2018regulation}. The potential of control theoretical approaches to and from neuromorphic engineering is more and more clear~\cite{sepulchre2019control,petri2024analysis,huijzer2025modelling,shahhosseini2024operator,schmetterling2024,che2023dominant,juarez2025collective,cathcart2024spiking}. Because events are {\it discrete} entities evolving in {\it continuous} time, neuromorphic control systems have the potential to inherit the best of both analog/continuous-time and digital/discrete-time control approaches~\cite{sepulchre2022spiking,sepulchre2019control}. Furthermore, many existing control strategies (like event-based~\cite{astrom2008event} and event-triggered control~\cite{heemels2012introduction}, maximum hands-off control~\cite{nagahara2016}, sliding-mode control~\cite{shtessel2014sliding}, and hybrid systems theory~\cite{goebel}) share at least some similarities with the neuromorphic approach and provide sets of mature tools for the analysis and the design of neuromorphic control systems. 

However, what a neuromorphic control system is remains unclear. As opposed to continuous-time, digital, event-based, or hybrid systems, we still miss a general definition of ``neuromorphic control systems''. The lack of such a definition can be detrimental to neuromorphic control theory by encouraging piecemeal and isolated, instead of more holistic and collective, rigorous theoretical developments and approaches to this new discipline and class of systems.

Motivated by these observations, we build upon recent efforts aiming at defining and characterizing neuromorphic control systems~\cite{fernandez2023neuromorphic,schmetterling2024} to propose a general modeling architecture for neuromorphic control systems. We rigorously define all the elements of the proposed architecture in such a way that they respect the constraints imposed by neuromorphic event-based hardware. We further propose a general, purely input-output way of formulating neuromorphic control problems in a way that conforms to the distinctive properties of neuromorphic sensing, computing, and actuation. We then restrict our attention to a class of rhythmic neuromorphic control problems and show how, in this setting, the proposed modeling framework can make the analysis of the resulting neuromorphic closed-loop systems amenable to a number of mature control-theoretical tools. Similarly to~\cite{fernandez2023neuromorphic,schmetterling2024}, the proposed methods are applied and developed in detail on the neuromorphic pendulum control problem, for which we rigorously establish error bounds depending on the system and control parameters. Interestingly, the closed-loop system exhibits a discrete-time Hopf-like bifurcation. We formulate a robust optimization problem on the error and show that, remarkably, the optimal solution lies precisely at the bifurcation point. The complementary paper~\cite{petri2025} engages the same problem from the hybrid system perspective. It develops a hybrid model of neuromorphic pendulum control, shows the existence and uniqueness of a hybrid limit cycle, and certifies robustness by establishing a uniform exponential stability property. Both papers provide complementary insights and results that are of independent interest.

\section{GENERAL ARCHITECTURE}

We propose the feedback control architecture in Fig.~\ref{fig:arch} as the fundamental and most basic representative of a neuromorphic control system. The plant is assumed to be a single-input and single-output linear time-invariant system, characterized by a transfer function $P$ in the field of rational functions with real coefficients $\RR(s)$.

%After set-point regulation, generating periodic motions is probably the most common control task. Based on the neuromorphic benchmark problem proposed in~\cite{schmetterling2024}, we will set forward a general architecture for achieving the latter goal. Then, we will propose a modeling framework to analyze and design the corresponding neuromorphic controller.

\begin{figure}[thpb]
    \centering
    \includegraphics[width=\columnwidth]{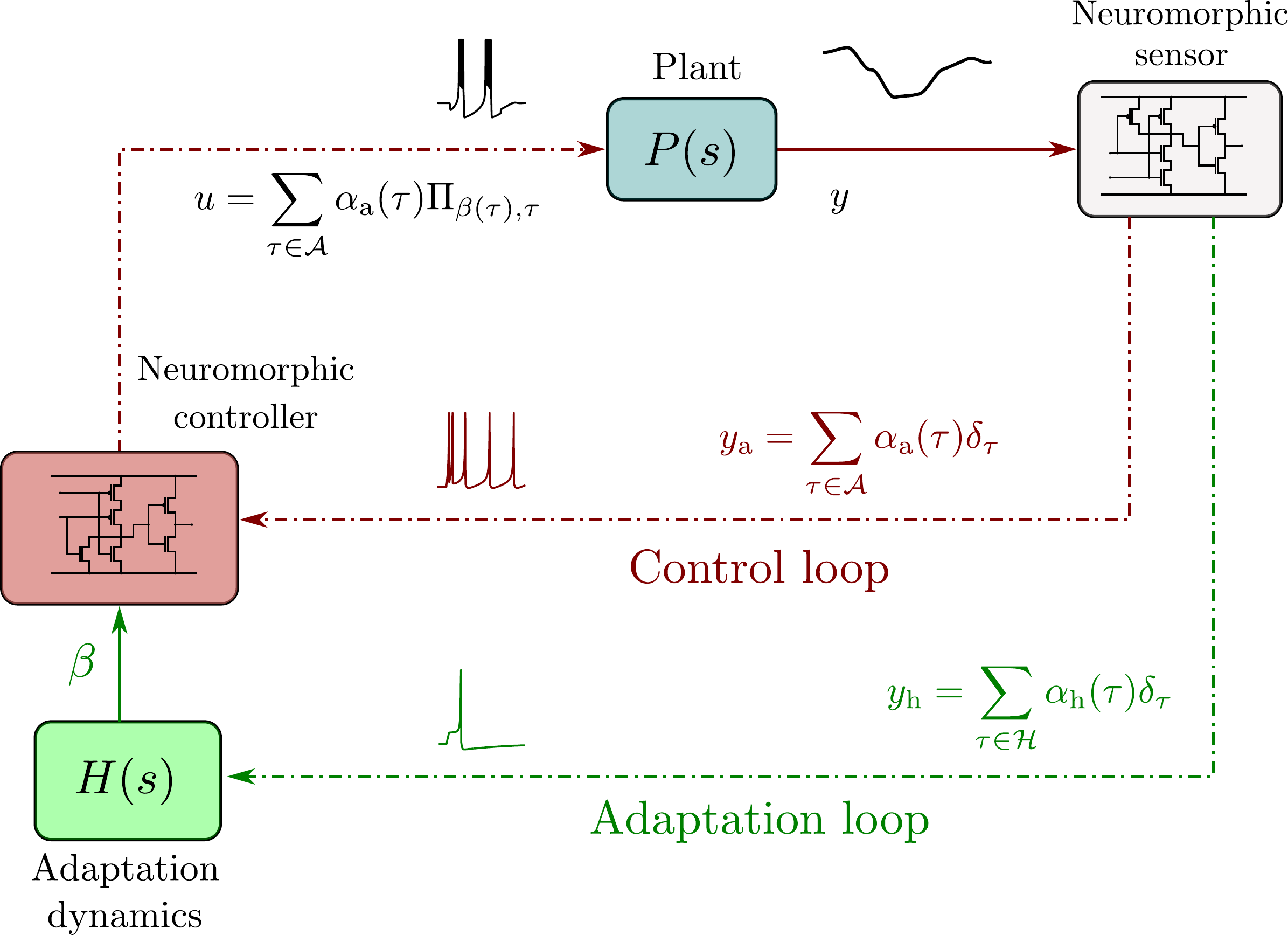}
    \caption{Neuromorphic-control architecture.}% \AF{I think we should add a Neuromorphic Actuator block, which works almost as a dual of the neuromorphic sensor: it receives events and generate analog signals compatible with the used physical actuators.}}
    \label{fig:arch}
\end{figure}

\subsection{Neuromorphic Sensors}

% \AF{(AF: we need to add the reference part in the sensor, including in the figure)}
% \textcolor{fernando}{Even in the general setting, i.e., when the objective is not necessarily a rhythmic behavior?}
% \AF{(AF: I cannot imagine a control problem without a reference... Lert's see how the problem formulation section looks like and decide after that. Starting on that now.)}
% \textcolor{fernando}{The original figure had $A^\star$ as reference, but I removed it for generality. Maybe I can just add something generic like 'ref'?}

The sensors of a neuromorphic control system are distinctly different as compared to both analog, continuous-time sensors and digital, clocked sensors. A neuromorphic sensor is made of three basic elements: an analog front-end that receives a transduced analog signal of the plant output; a ``threshold'' element that triggers an all-or-none response when specific conditions or behaviors are met or detected in the plant output; an event-generator that
%processes and amplifies, possibly nonlinearly, the threshold all-or-none response and
transmits the all-or-none response of the threshold element to the rest of the system as ``impulses'' (usually transmitted as few-bit events through an asynchronous digital channel).
%The dynamic vision sensor at the core of event-based cameras is probably the best understood and most widespread neuromorphic sensors.

Sensory events are modeled as Dirac delta distributions. More precisely, a sensory event generated at time $\tau$ is represented by the, possibly weighted, Dirac delta distribution $\delta_\tau$ satisfying $\int_{-\infty}^\infty f(t)\delta_{\tau}(t)\rd t = f(\tau)$ for all continuous functions $f : \RR \to \RR$.

Note that the same plant output can be processed, encoded, and transmitted by different neuromorphic sensors. The threshold condition that has to be met to generate a sensory event is what distinguishes two neuromorphic sensors encoding the same signal. Threshold-based neuromorphic sensing is intrinsically semantic, in the sense that each event generated by the sensor has a well-defined meaning, e.g., ``luminosity is increasing'' or ``the maximum has been reached''. This is in sharp contrast to both analog and digital sensing, which mostly aims at transducing the value of a physical variable of interest into an analog or a digital signal. Threshold-based neuromorphic sensing is also fundamentally different to event-based sampling~\cite{astrom2008event,heemels2012introduction} and $\Delta/\Delta\Sigma$-modulation~\cite{Sira-Ramirez2015,razavi2016delta}.

In the simple architecture of Fig.~\ref{fig:arch} there are two kinds of sensory events: \emph{actuation events},  meant to trigger the activation of the neuromorphic control loop actuator, and \emph{adaptation events}, meant to trigger an impulsive response in the adaptation block, as determined by its transfer function $H \in \RR(s)$.
%A neuromorphic sensor monitors the plant's output and detects events. There are two types of events: \emph{actuation events}, which are designated for the neuromorphic controller, and \emph{adaptation events}, designated for an adaptive unit with transfer function $H \in \RR(s)$.

Let $\AC \subset \RR$ be the set of instants of time at which the actuation events are generated. At each actuation event, a Dirac delta is sent from the sensor to the controller. The event-based input from the actuation sensor to the controller is 
%The occurrence of an event is signaled to the controller by generating a spike.
%A Dirac delta distribution approximates each spike, so the signal traveling from the sensor to the controller consists of a train of spikes
\begin{displaymath}
    y_\ra(t) = \sum_{\tau \in \AC}\alpha_\ra(\tau)\delta_{\tau}(t) \;.
\end{displaymath}
% where, by definition, $\delta_{\tau}$ satisfies
% \begin{displaymath}
%     \int_{-\infty}^\infty f(t)\delta_{\tau}(t)\rd t = f(\tau)
% \end{displaymath}
% for all continuous functions $f : \RR \to \RR$.
The function $\alpha_\ra : \AC \to \left\{+1,-1\right\}$ is used to attach to each event a 1-bit information content that can be used to define and distinguish \emph{positive} and \emph{negative} actuation events.

Similarly, let $\HH \subset \RR$ be the set of instants of time at which the adaptation events are generated. At each adaptation event, a Dirac delta is sent from the sensor to the adaptation block. The event-based input from the adaptation sensor to the adaptation block is
%Similarly, the occurrence of an event is signaled to the adaptive unit by generating a spike, so the signal from the sensor to the controller also consists of a train of spikes
\begin{displaymath}
    y_\rh(t) = \sum_{\tau \in \HH}\alpha_\rh(\tau)\delta_{\tau}(t) \;,
\end{displaymath}
where, as before, $\alpha_\rh : \HH \to \{+1,-1\}$ serves to distinguish \emph{positive} and \emph{negative} events.

\subsection{Neuromorphic actuation}

Upon reception of a signed spike at time $\tau \in \AC$, the controller fires a burst of spikes that lasts $\beta(\tau) > 0$ units of time and of the same sign as the received spike. Since the inter-burst period (i.e., the dwel time between two bursts) is designed to be much shorter than the dominant time constant of $P$, a burst can be reasonably modeled by a finite pulse
\begin{displaymath}
    \Pi_{\beta,\tau}(t) =
    \begin{cases}
        1 & \text{if $\tau \le t \le \tau + \beta$} \\
        0 & \text{otherwise}
    \end{cases} \;,
\end{displaymath}
so that the control signal takes the form
\begin{equation} \label{eq:u}
    u(t) = \sum_{\tau \in \AC}\alpha_\ra(\tau)\Pi_{\beta(\tau),\tau}(t) \;.
\end{equation}

The choice of using bursting, instead of spiking, neurons for the actuation stage stems mainly from two reasons. Single, isolated spikes, are mostly filtered out by the low-pass nature of the actuators. Conversely, a burst of spikes received in rapid succession is integrated by the actuator dynamics, leading to reliable responses. In biological nervous systems, most motor neurons are indeed busters. Also, as opposed to single all-or-none spikes, bursts also possess a graded nature. Increasing (decreasing) the burst duration or the number of spikes per burst leads to stronger (weaker) actuation, which enables more finely tuned control actions.

% ... it is not a coincidence that most motor neurons are bursters ...

% \textcolor{fernando}{Alessio: Why bursts? Because they enjoy the advantages of spikes while being easy to modulate?}

\subsection{Neuromorphic adaptation}

Assuming that the adaptive unit starts at rest, the burst width is given by
\begin{equation} \label{eq:beta}
    \beta(t) = \sum_{\tau \in \HH}\alpha_\rh(\tau)h(t-\tau) \;,
\end{equation}
where $h : \RR \to \RR$ is the impulse response of the adaptation unit, i.e., $h(t) = \mathcal{L}^{-1}\{H(s)\}$ with $\mathcal{L}$ the Laplace transform. The causality of the adaptive unit implies that $h(t-\tau) = 0$ for $t \le \tau$. Thus, the expression~\eqref{eq:beta} reduces to
\begin{equation} \label{eq:beta_causal}
    \beta(t) = \sum_{\substack{\tau \in \HH \\ \tau < t}}\alpha_\rh(\tau)h(t-\tau) \;.
\end{equation}

To summarize, the designer parameters are the transfer function $H$, the event generation rule fo the sets $\AC$ and $\HH$, and the functions $\alpha_\ra$ and $\alpha_\rh$.

\subsection{Neuromorphic control problem formulation}

The sparse and impulsive temporal nature of neuromorphic sensors and actuators calls for new ways of formalizing control objectives in the neuromorphic setting. 

Consider, for instance, the problem of generating a stable periodic oscillatory behavior with a desired amplitude $A^\star$ in a controlled plant. A common approach~\cite{pasandi2022integrated,ijspeert2007swimming} to formalize such a problem is to generate a reference signal $y^\star(t)$ such that $\max_{t\geq0}|y^\star(t)|=A^\star$ and then to apply, e.g., output regulation or tracking techniques to ensure that $y(t)$ asymptotically converges to $y^\star(t)$. Such an approach does not naturally generalize to the neuromorphic setting, where sensors might only be able to generate 1-bit events and only when the output of the system crosses specific thresholds, which makes $y(t)$ fundamentally unmeasurable.

Alternatively, consider $m$ input-output (non-necessarily causal) continuous operators
\begin{equation*}
    \mathcal{E}_{i}: \mathcal{L}_{q_ie}(\RR_{\geq 0},\RR) \to \mathcal{L}_{q_ie}(\RR_{\geq 0},\RR)\,,
\end{equation*}
where $q_1,\ldots,q_m\in\mathbb N$ and $\mathcal{L}_{q_ie}(\RR_{\geq t_0},\RR)$ denotes the extended $\mathcal{L}_{q_i}(\RR_{\geq t_0},\RR)$ space. A neuromorphic control problem can then be formalized as follows. Given a performance criterion $\varepsilon\geq 0$, design neuromorphic sensors, controller, and adaptation dynamics for the architecture in Fig.~\ref{fig:arch} guaranteeing that 
    \begin{equation}\label{eq: general neuromorphic}
        \limsup_{t\to\infty} \big|\mathcal{E}_{i}(y)(t)\big|\leq \varepsilon,\quad\text{for all $i\in\{1,\ldots,m\}$} \,.
    \end{equation}
The choice of the operators $\mathcal{E}_{i}$ is key. For instance, the problem of generating a stable periodic oscillatory behavior with a desired amplitude $A^\star$ in a controlled plant can be formalized as in~\eqref{eq: general neuromorphic} with $m=2$, $q_1 = 2$, $q_2 = \infty$, and
\begin{subequations}\label{eq:rhythm control problem}
    \begin{align}
%    &\mathcal{E}_{1}:\mathcal{L}_{2e}(\RR_{\geq 0},\RR) \to \mathcal{L}_{2e}(\RR_{\geq 0},\RR)\nonumber\\
    &\mathcal{E}_{1}(y)(t)=y(t+2\pi/\omega) - y(t) \,, \label{eq:periodic} \\ %[0.2cm]
%    &\mathcal{E}_{2}:\mathcal{L}_{\infty}(\RR_{\geq 0},\RR) \to \mathcal{L}_{\infty}(\RR_{\geq 0},\RR)\nonumber\\
    &\mathcal{E}_{2}(y)(t)= A^\star - \max_{\tau\geq t}|y(\tau)| \,,
\end{align}
\end{subequations}
where, in~\eqref{eq:periodic}, $\omega>0$ is a free design parameter. 

Such a formalization constitutes a semantic, instead of representational, approach to event-based sensing and control system design. Instead of representing signals (e.g., deviation from a reference) in some absolute, quantitative format suitable for algorithmic processing, each event can bring rich, context-dependent meaning that enables fast and adaptive decision-making about what and how to sense, and when and how to act. For instance, a 1-bit event is sufficient to bring the message ``the sensor signal reached a local maximum; at peak it was above (below) the reference level''. This approach can foster the design of neuromorphic sensors and controllers that use sparser, but semantically richer, event-based signals, as thoroughly illustrated in the remainder of the paper on the rhythmic control problem defined by~\eqref{eq:rhythm control problem}.

%The control system is depicted in Fig.~\ref{fig:arch}.

% \AF{
% I have the impression we can formulate something radically more general and then specialize to the limit cycle case. For instance:\\

% Let $\mathcal E:C_0(\RR_{\geq0})\to C_0(\RR_{\geq0})$ be a continuous operator. The goal is to design a neuromorphic control law such that, asymptotically, the plant's output $y(t)$ satisfies $\mathcal E(y)\equiv 0$.

% Example: $$\mathcal E(f)=\left(f(t+2\pi/\omega)-f(t)\right)^2 + \left(\max_{\tau \ge 0}|f(\tau)| - A^\star\right)^2$$.

% }

% {\color{fernando}
%  It may be easier to deal with something like $\mathcal{E} : \mathcal{L}_2 \to \RR$,
%  \begin{multline*}
%      \mathcal{E}(f) = \int_0^\infty \left( f(\tau+2\pi/\omega)-f(\tau) \right)^2 \rd \tau + \\ \left(\max_{\tau \ge 0}|f(\tau)| - A^\star\right)^2 \;.
%  \end{multline*}
% }

%%%%%%%%%%%%%%%%%%%%%%%%%%%%%%%%%%%%%%%%%%%%%%%%%%%%%%%%%%%%%%%%%%%%%%%%%%%%%%%%
\section{RHYTHMIC NEUROMORPHIC CONTROL DESIGN}

In this section, we set to solve the neuromorphic control problem defined by~\eqref{eq:rhythm control problem}. In particular, we outline the main assumptions, establish the general design procedure, and justify the choices made for some of the design parameters.

%Observe that, given a desired error tolerance $\varepsilon$, the neuromorphic control problem defined by~\eqref{eq:rhythm control problem} can be explicitly rephrased as designing a neuromorphic control architecture as in Figure~\ref{fig:arch} such that the plant output satisfies
%\begin{align*}
%    \limsup_{t\to\infty} |y(t+2\pi/\omega) - y(t)|\leq\varepsilon\\
%    \limsup_{t\to\infty} \left|\max_{t \ge 0}|y(t)| - A^\star \right|\leq\varepsilon
%\end{align*}
%for some $\omega > 0$. In this section, we outline the main assumptions, establish the general design procedure, and justify the choices made for some of the design parameters in order to solve this neuromorphic control problem.

% \paragraph*{Problem Statement} Similar to the problem statement in~\cite{schmetterling2024}, the control objective consists of robustly stabilizing a periodic solution of amplitude $A^\star > 0$. That is, a solution satisfying
% \begin{displaymath}
%     y(t+2\pi/\omega) = y(t) \quad \text{and} \quad \max_{t \ge 0}|y(t)| = A^\star
% \end{displaymath}
% for some $\omega > 0$.

%Of course, there are many ways to choose the design parameters to achieve the control objective.
One of the main challenges in analyzing the system shown in Fig.~\ref{fig:arch} is that the impulses in $y_\rh$ cause jumps in $\beta$, effectively creating a hybrid system like those discussed in~\cite{goebel}. We approach the hybrid nature of the problem from a singular perturbation perspective (see, e.g.,~\cite{kokotovic,khalil}). We argue that, by making the adaptation dynamics slow enough, a sufficient timescale separation exists so that the dynamics on the fast timescale are purely continuous. In contrast, those on the slow timescale are purely discrete.

The rationale is the following:
\begin{enumerate}
    \item Assume that, compared with the dynamics of $P$ in closed loop with the neuromorphic controller (red loop in Fig.~\ref{fig:arch}), the dynamics of the adaptive unit $H$ are sufficiently slow so that $\beta$ can be considered constant on a fast timescale. \label{it:time_scale}
    \item Apply the describing-function method~\cite{gelb,khalil} to establish the existence of a limit cycle of amplitude $A > 0$ on the fast timescale.
    \item Design $H$ so that $A \to A^\star$ on the slow timescale and such that the assumption in Item~\ref{it:time_scale} holds. By the discontinuous nature of $\beta$, the slow dynamics are easier to analyze in discrete time.
\end{enumerate}

Let us to elaborate on the last two steps.

\subsection{Fast dynamics}

In the fast timescale, $\beta$ remains constant. Using the describing-function method~\cite{gelb,khalil}, we investigate the potential existence of limit cycles and estimate their amplitude and frequency as functions of $\beta$. 

We begin by narrowing down the actuation events.

\begin{assumption} \label{ass:actuation_events}
    We have $\AC = \left\{ \tau \in \RR \mid g_\ra(y(\tau),\dot{y}(\tau)) = 0 \right\}$
    for some function $g_\ra : \RR \times \RR \to \RR$ satisfying
    \begin{subequations} \label{eq:symmetry}
    \begin{equation}
        g_\ra(y,\dot{y}) = 0 \quad \Longleftrightarrow \quad g_\ra(-y,-\dot{y}) = 0 
    \end{equation}
    and such that $y(t) = A\sin(\omega t)$, with $A,\omega > 0$, implies that $\AC$ is at most countable. Also,
    $\alpha_\ra(\tau) = \hat{\alpha}_\ra(y(\tau),\dot{y}(\tau))$, $\tau \in \AC$, for some function $\hat{\alpha}_\ra : \RR \times \RR \to \RR$ satisfying
    \begin{equation}
        \hat{\alpha}_\ra(-y,-\dot{y}) = -\hat{\alpha}_\ra(y,\dot{y}) \;.
    \end{equation}
    \end{subequations}
\end{assumption}

The conditions on $g_\ra$ and $\hat{\alpha}_\ra$ are analog to the condition of a nonlinearity being \emph{odd} (see~\cite{gelb}). They can be relaxed at the expense of additional notation and computational burden.

\begin{proposition}
    Consider the map $y \mapsto u$ defined by~\eqref{eq:u} with $\beta$ constant and $\AC$ and $\alpha_\ra$ satisfying Assumption~\ref{ass:actuation_events}. The map has a well-defined describing function $N_\beta(A,\omega)$.
\end{proposition}

\begin{proof}
    Suppose that $y(t) = A\sin(\omega t)$. The control~\eqref{eq:u} is periodic with frequency $\omega$. This follows from the fact that $g_\ra$ is static, time-invariant, and depends only on the periodic function $y$ and its derivative. Thus,
    \begin{displaymath}
        u(t) = a_0 + \sum_{n=1}^\infty \big(a_n\cos\left(\omega n t\right)+b_n\sin \left(\omega n t \right)\big) \;.
    \end{displaymath}
    %with
    %\begin{subequations} \label{eq:fourier}
    %\begin{align}
    %    a_n &= \frac{\omega}{\pi} \int_0^{2\pi/\omega} u(t)\cos(\omega n t) \rd t \\
    %    b_n &= \frac{\omega}{\pi} \int_0^{2\pi/\omega} u(t)\sin(\omega n t) \rd t \;,
    %\end{align}
    %\end{subequations}
    %$n \neq 0$.
    By the symmetry condition~\eqref{eq:symmetry}, the average of $u$ on each period is zero, so $a_0 = 0$. Hence, the describing function~\cite[Eqs.~(2.2-16), (2.2-28)-(2.2-34)]{gelb} is
    \begin{displaymath}
        N_\beta(A,\omega) = \frac{b_1+ja_1}{A} \;,
    \end{displaymath}
    where $a_1$ and $b_1$ are the first Fourier coefficients depending on $A$, $\omega$, and $\beta$ implicitly.
\end{proof}

Suppose that the harmonic balance (HB) equation
\begin{equation} \label{eq:harmonic}
    N_\beta(A,\omega)P(j\omega) - 1 = 0 
\end{equation}
is satisfied for some frequency $\omega$ and amplitude $A$, both depending on $\beta$. Then, a limit cycle of amplitude $A$ and frequency $\omega$ is established, provided that $P$ has a low-pass characteristics of a filter with negligeable output for frequencies above $\omega$ (see~\cite{gelb} for details). We assume that such a limit cycle is stable.

\subsection{Slow dynamics}

In the slow timescale, the state trajectory is assumed to have already converged to a periodic solution whose amplitude and frequency are implicitly determined by~\eqref{eq:harmonic}. To emphasize the dependence on $\beta$, we write $A = \hat{A}(\beta)$ and $\omega = \hat{\omega}(\beta)$. The purpose of the adaptation loop is to let $\hat{A}(\beta)$ approach the desired amplitude $A^\star$.

Let $t_k \in \AC$ be the time instant corresponding to the $k$th-actuation event. Evaluating~\eqref{eq:beta_causal} at $t = t_{k+1}$ gives
\begin{equation} \label{eq:beta_t_a}
    \beta\left(t_{k+1}\right) = \sum_{\substack{\tau \in \HH \\ \tau < t_{k+1}}}\alpha_\rh(\tau)h\left(t_{k+1}-\tau\right) \;.
\end{equation}
The difference between the actual and desired amplitude is the relevant information to be transmitted along the adaptive loop (green path in Fig.~\ref{fig:arch}). Note also that the amplitude of the limit cycle can effectively be measured when $\dot y=0$.

\begin{assumption} \label{ass:adaptation_events}
    We have $\HH = \left\{ \tau \in \RR \mid \dot{y}(\tau) = 0 \right\}$ and
    \begin{equation} \label{eq:alpha_h}
        \alpha_\rh(\tau) \in \sgn(A^\star-y(\tau)) \;, \quad \tau \in \HH \;,
    \end{equation}
    where $\sgn$ is the multivalued \emph{signum} function: 
    \begin{displaymath}
        \sgn(x) = 
        \begin{cases}
            \{-1\} & \text{if $x < 0$} \\
            [-1,1] & \text{if $x = 0$} \\
             \{1\} & \text{if $x > 0$}
        \end{cases} \;.
    \end{displaymath}
\end{assumption}

Given an adaptation event time $\tau \in \HH$, define $\rho_\tau$ as the preceding actuation event time, $\rho_\tau = \max\{r \in \AC \mid r \le \tau\}$. By substituting~\eqref{eq:alpha_h} in~\eqref{eq:beta_t_a} and noting that $y(\tau) = \hat{A}(\beta(\rho_\tau))$, %\textcolor{taya}{($y(\tau_{k+1}) = A(\beta(t_k))$ where $\tau_{k+1} \in \HH$, $t_k \in \AC$ and $t_k < \tau_{k+1}$)} 
we obtain
\begin{equation} \label{eq:beta_t_a_2}
    \beta\left(\tau_{k+1}\right) \in \sum_{\substack{\tau \in \HH \\ \tau < \tau_{k+1}}}\sgn\Big(A^\star-\hat{A}\big(\beta(\rho_\tau)\big)\Big)h\left(\tau_{k+1}-\tau\right) \;.
\end{equation}

We consider a first-order low-pass transfer function for the adaptive unit. This will considerably simplify~\eqref{eq:beta_t_a_2}.

\begin{assumption} \label{ass:adaptive}
    The transfer function of the adaptive unit is
    \begin{equation} \label{eq:low_pass}
        H(s) = \frac{\gamma}{s+c} \;, \quad \gamma, c > 0 \;.
    \end{equation}
\end{assumption}

Note that the inverse Laplace transform of~\eqref{eq:low_pass} is
\begin{displaymath}
    h(t) = 
    \begin{cases}
                      0 & \text{if $t \le 0$} \\
        \gamma e^{-c t} & \text{if $t > 0$} 
    \end{cases} \;.
\end{displaymath}
This impulse response has a simple but important property:
\begin{equation} \label{eq:group_h}
    h(t_1 + t_2) = \frac{1}{\gamma}h(t_1)h(t_2) \;, \quad t_1, t_2 \ge 0 \;,
\end{equation}
indeed,
\begin{equation} \label{eq:h_1_order}
    \gamma e^{-c(t_1 + t_2)} = \gamma e^{-c t_1}e^{-c t_2} \;.
\end{equation}

\begin{proposition} \label{prop:recursion}
    Under Assumptions~\ref{ass:adaptation_events} and~\ref{ass:adaptive}, the expression~\eqref{eq:beta_t_a} can be written as the recursion
    \begin{multline} \label{eq:beta_recursion}
        \beta\left(t_{k+1}\right) \in \frac{h\left(\Delta t_k\right)}{\gamma}\beta\left(t_{k}\right) + \\
            \sum_{\substack{\tau \in \HH \\ t_{k} \le \tau < t_{k+1}}}\sgn\Big(A^\star-\hat{A}\big(\beta(\rho_\tau)\big)\Big)h\left(t_{k+1}-\tau\right) \;.
    \end{multline}
    % \textcolor{taya}{
    %     \begin{multline} \label{eq:beta_recursion}
    %         \beta\left(t_{k+1}\right) \in \frac{h\left(\Delta t_k\right)}{\gamma}\beta\left(t_{k}\right) + \\
    %             \sum_{\substack{\tau \in \HH \\ t \in \AC \\ t_{k} \le t < \tau < t_{k+1}}}\sgn\big(A^\star-A\left(\beta\left(t\right)\right)\big)h\left(t_{k+1}-\tau\right) \;.
    %     \end{multline}
    % }
\end{proposition}

\begin{proof}
    Define $\Delta t_k = t_{k+1} - t_{k}$ and let us split the sum in~\eqref{eq:beta_t_a_2} as
    \begin{multline}
        \beta\left(t_{k+1}\right) \in 
            \sum_{\substack{\tau \in \HH \\ \tau < t_{k}}}\sgn\Big(A^\star-\hat{A}\big(\beta(\rho_\tau)\big)\Big)h\left(\Delta t_k + t_{k} - \tau\right) + \\
            \sum_{\substack{\tau \in \HH \\ t_{k} \le \tau < t_{k+1}}}\sgn\Big(A^\star-\hat{A}\big(\beta(\rho_\tau)\big)\Big)h\left(t_{k+1}-\tau\right) \;.
    \end{multline}
    By~\eqref{eq:group_h} and~\eqref{eq:beta_t_a_2} we recover~\eqref{eq:beta_recursion}.
\end{proof}

\begin{corollary} \label{cor:recusion}
    Suppose that, in addition to the conditions of Proposition~\ref{prop:recursion}, the function $\hat{A}$ is monotonically increasing and there exists a burst width $\beta^\star \ge 0$ such that $\hat{A}(\beta^\star) = A^\star$. Then, the recursion~\eqref{eq:beta_recursion} simplifies to
    \begin{multline} \label{eq:beta_recursion_2}
        \beta\left(t_{k+1}\right) \in \frac{h\left(\Delta t_k\right)}{\gamma}\beta\left(t_{k}\right) + \\
            \sum_{\substack{\tau \in \HH \\ t_{k} \le \tau < t_{k+1}}}\sgn\left(\beta^\star-\beta\left(\rho_\tau\right)\right)h\left(t_{k+1}-\tau\right) \;.
    \end{multline}
\end{corollary}

Under the assumptions of the corollary, our problem boils down to showing that, in some sense to be clarified below, $\beta\to\beta^\star$. The only remaining degrees of freedom to achieve this are the constants $c$ and $\gamma$, and the functions $\hat{\alpha}_\ra$ and $g_\ra$, whose choice depends in general on the plant $P$.

\section{CASE STUDY: OSCILLATIONS OF DESIRED AMPLITUDE ON A PENDULUM}

As in~\cite{schmetterling2024,fernandez2023neuromorphic}, we now focus on the controlled pendulum
\begin{equation} \label{eq:nonlinear}
    \ddot{y}(t) + 2\xi \omega_\rn \dot{y}(t) + \omega_\rn^2 \sin(y(t)) = \lambda u(t) \;, 
\end{equation}
where $\xi \in [0, 1]$ is the damping ratio, $\omega_\rn > 0$ the undamped natural frequency, and $\lambda > 0$ the input gain. The transfer function of the linearization of~\eqref{eq:nonlinear} is $P(s) = \lambda/(s^2 + 2\xi \omega_\rn s + \omega_\rn^2)$.

\subsection{Fast dynamics, self-sustained oscillations}

A simple choice for the functions $g_\ra$ and $\hat{\alpha}_\ra$, consistent with Assumption~\ref{ass:actuation_events}, is $g_\ra(y,\dot{y}) = y$ and $\hat{\alpha}_\ra(y,\dot{y}) \in \sgn(\dot{y})$. We compute the controller describing function $N_\beta(A,\omega)$ in the standard way~\cite{gelb}. Assume that $y(t) = A\sin(\omega t)$, $t \ge 0$. During one period, the actuation events are $\AC \cap [0,2\pi/\omega) = \left\{0, \pi/\omega\right\}$. Hence, for $\beta < \pi/\omega$, the control is $u(t) = \Pi_{\beta,0}(t) - \Pi_{\beta,\pi/\omega}(t)$, $t \in [0,2\pi/\omega)$. The coefficients of the first harmonic are $a_1 = \frac{2}{\pi}\sin(\omega \beta)$ and $b_1 = \frac{2}{\pi}\left(1 - \cos(\omega \beta)\right)$.
%\begin{align*}
%    a_1 &= \frac{2\omega}{\pi} \int_0^{\beta} \cos(\omega t) \rd t = \frac{2}{\pi}\sin(\omega \beta) \\
%    b_1 &= \frac{2\omega}{\pi} \int_0^{\beta} \sin(\omega t) \rd t = \frac{2}{\pi}\left(1 - \cos(\omega \beta)\right) \;.
%\end{align*}
It follows that
\begin{subequations} \label{eq:N}
\begin{equation} \label{eq:N_amp}
    |N_\beta(A,\omega)| = \frac{\sqrt{a_1^2+b_1^2}}{A} = \frac{4}{A \pi}\sin\left(\frac{\omega \beta}{2}\right) \;,
\end{equation}
and
\begin{equation} \label{eq:N_ang}
    \angle N_\beta(A,\omega) = \arctan \frac{a_1}{b_1} = \frac{\pi-\omega \beta}{2} \;. 
\end{equation}
\end{subequations}

\begin{proposition} \label{prop:unique}
    For every $\beta > 0$, there exists a unique amplitude $A$ and frequency $\omega$ satisfying the HB equation~\eqref{eq:harmonic}.
\end{proposition}

\begin{proof}
The HB equation~\eqref{eq:harmonic} is equivalent to
\begin{equation} \label{eq:HB}
    |N_\beta(A,\omega)|=\frac{1}{|P(j\omega)|} \quad \text{and} \quad \angle N_\beta(A,\omega) = -\angle P(j\omega) \;.   
\end{equation}
Note that $-\angle P(j\omega)$ is a strictly monotonically increasing function. On the other hand, we can see from~\eqref{eq:N_amp} that $\angle N_\beta(A,\omega)$ is a strictly monotonically decreasing function. Moreover,
$\angle P(0) < \angle N_\beta(A,0)$ and $\angle N_\beta(A,\pi/\beta) < -\angle P(j\pi/\beta)$.
%\begin{align*}
%                    0 = \angle P(0) \; &< \; \angle N_\beta(A,0) = \pi/2  \\
%    0 = \angle N_\beta(A,\pi/\beta) \; &< \; -\angle P(j\pi/\beta) \;.
%\end{align*}
Thus, there exists a unique $\omega \in (0,\pi/\beta)$ satisfying~\eqref{eq:HB}. Let $\hat{\omega}(\beta)$ be the function that assigns this value to $\omega$ for each $\beta$. It follows from~\eqref{eq:N_amp} that
\begin{equation}\label{eq:amp_of_beta}
   A =  \hat A(\beta) = \frac{4}{\pi}\left|P\big(j\hat{\omega}(\beta)\big)\right|\sin\left(\frac{\hat{\omega}(\beta) \beta}{2}\right) \;.
\end{equation}
\end{proof}

\begin{figure}
    \centering
    \subfigure[Frequency]{\includegraphics[width=0.8\columnwidth]{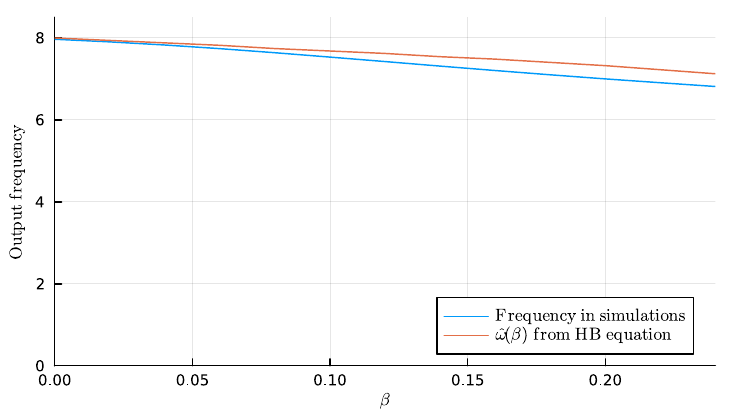}}
    \subfigure[Amplitude]{\includegraphics[width=0.8\columnwidth]{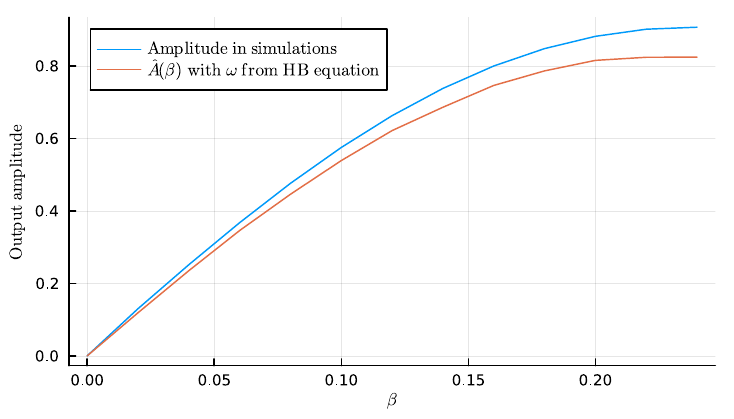}}
    \caption{Frequency and amplitude for $\beta \in [0, 0.2]$ obtained using the HB equation (orange line) and by simulation (blue line). The pendulum parameters are $\lambda = 15$, $\xi = 0.1$, and $\omega_\rn = 8$.} 
    \label{fig:freq_ampl}
\end{figure}

Figure~\ref{fig:freq_ampl} compares the amplitude and frequency obtained by simulation with those obtained by~\eqref{eq:HB}. % The parameters of the pendulum were
%\begin{equation} \label{eq:pars}
 %    $\lambda = 15$, $\xi = 0.1$, and $\omega_\rn = 8$.
%\end{equation}

\subsection{Slow dynamics, discrete-time behavior}

We can see from Fig.~\ref{fig:freq_ampl} that $\hat{A}$ is monotonically increasing, so Corollary~\ref{cor:recusion} holds and~\eqref{eq:beta_recursion_2} reduces to
\begin{equation} \label{eq:beta_recursion_3}
    \beta\left(t_{k+1}\right) \in \frac{h\left(\Delta t_k\right)}{\gamma}\beta\left(t_{k}\right) +
        \sgn\big(\beta^\star-\beta\left(\rho_\tau\right)\big)h\left(t_{k+1}-\tau\right)
\end{equation}
with $\tau$ the unique time in $\HH$ such that $t_{k} \le \tau < t_{k+1}$. Moreover, the time difference between two consecutive actuation events is half a period. In contrast, the time interval between an actuation and the following adaptation events is a quarter of the period. We will approximate the period by $2\pi/\omega^\star$ with $\omega^\star = \hat{\omega}(\beta^\star)$, that is, $\Delta t_k = \frac{\pi}{\omega^\star}$, so that $h\left(\Delta t_k\right) = \gamma e^{-c\frac{\pi}{\omega^\star}}$ and $h\left(t_k-\tau\right) = \gamma e^{-c\frac{\pi}{2\omega^\star}}$.
Finally, note that $\rho_\tau = t_k$ for $t_{k} \le \tau < t_{k+1}$, so that
\begin{equation} \label{eq:beta_first}
    \beta_{k+1} \in e^{-c\frac{\pi}{\omega^\star}} \beta_k - \gamma \sgn\left(\beta_k-\beta^\star\right)e^{-c\frac{\pi}{2\omega^\star}}\,,
\end{equation}
where, for notational simplicity, we let $\beta_k = \beta(t_k)$.

Our task now is to study the dynamics of the error $\tilde{\beta}_k = \beta_k - \beta^\star$. It follows from~\eqref{eq:beta_first} that
\begin{equation} \label{eq:beta_error}
    \tilde{\beta}_{k+1} \in -g_0 + g_1 \tilde{\beta}_k - g_2 \sgn(\tilde{\beta}_k)
\end{equation}
with
\begin{equation} \label{eq:gs}
    g_0 = (1 - e^{-c\frac{\pi}{\omega^\star}}) \beta^\star > 0 \;, \quad g_1 = e^{-c\frac{\pi}{\omega^\star}} \in (0,1) \;, \quad
        \text{and} \quad  g_2 = \gamma e^{-c\frac{\pi}{2\omega^\star}} > 0\;.
\end{equation}
Let us first compute the fixed points of~\eqref{eq:beta_error}.

\begin{proposition}
    The difference inclusion~\eqref{eq:beta_error} has a unique fixed point
    \begin{equation} \label{eq:beta_fix}
        \tilde{\beta}_\fix = \min\left\{0,\frac{g_2-g_0}{1-g_1}\right\} \le 0 \;.
    \end{equation}
\end{proposition}

\begin{proof}
    The fixed points $\tilde{\beta}_\fix$ of~\eqref{eq:beta_error} are those satisfying $\tilde{\beta}_{\fix} \in -g_0 + g_1\tilde{\beta}_{\fix} - g_2 \sgn(\tilde{\beta}_\fix)$, that is,
    \begin{displaymath}
        -\frac{g_0}{1 - g_1} \in \tilde{\beta}_\fix + \frac{g_2}{1 - g_1} \sgn(\tilde{\beta}_\fix) \;.
    \end{displaymath}
    Since the multivalued operator $\sgn$ is maximally monotone, the resolvent $(I_\rd + \epsilon \sgn)^{-1}$ with $\epsilon > 0$ is single-valued and defined in all $\RR$ (see, e.g.,~\cite{bauschke}). Indeed, we have
    \begin{displaymath}
	   (I_\rd + \epsilon \sgn)^{-1}(x) = x - \min\left\{ |x|, \epsilon \right\} \sgn(x) \;.
    \end{displaymath}
    Thus, there is a unique fixed point
    \begin{displaymath}
        \tilde{\beta}_\fix = -\frac{g_0}{1-g_1} + \min\left\{ \frac{g_0}{1-g_1}, \frac{g_2}{1-g_1} \right\} \;,
    \end{displaymath}
    which readily simplifies to~\eqref{eq:beta_fix}.
\end{proof}

There are two qualitatively different asymptotic behaviors of~\eqref{eq:beta_error}. The first one is characterized by the following.
\begin{proposition} \label{prop:behav_1}
    Let $g_0 \geq g_2 > 0$ and $g_1 \in (0,1)$. The fixed point $\tilde{\beta}_\fix = -(g_0-g_2)/(1-g_1)$ of the difference inclusion~\eqref{eq:beta_error} is globally asymptotically stable.
\end{proposition}

\begin{proof}
    Define the Lyapunov function $V_k = |\tilde{\beta}_k - \tilde{\beta}_\fix|$ and the difference $\Delta V_k =  V_{k+1} - V_k$. Let $\RR_{-}$ be the set of negative real numbers. Note that, for $\tilde{\beta}_k \in \RR_{-}$, we have
    \begin{equation} \label{eq:D_lyap_1}
        \tilde{\beta}_{k+1} = g_1 \tilde{\beta}_k + g_2 - g_0 = g_1 (\tilde{\beta}_k - \beta_\fix) + \beta_\fix \;.
    \end{equation}
    Thus, $\Delta V_k = g_1 |\tilde{\beta}_k - \tilde{\beta}_\fix| - |\tilde{\beta}_k - \tilde{\beta}_\fix| < 0$. This proves the asymptotic stability of $\tilde{\beta}_\fix$ and the invariance of the open interval $(2\tilde{\beta}_\fix,0)$.

    Note that, by~\eqref{eq:D_lyap_1}, $\tilde{\beta}_k < \tilde{\beta}_\fix$ implies that $\tilde{\beta}_{k+1} < \tilde{\beta}_\fix$, so the interval $(-\infty,\tilde{\beta}_\fix)$ is also invariant. Since $\RR_- = (-\infty,\tilde{\beta}_\fix)\cup (2\tilde{\beta}_\fix,0)$, $\RR_-$ is invariant as well. This shows that a basin of attraction of the origin is $\RR_-$.

    Note that, since $g_0 \ge g_2$, we have
    \begin{displaymath}
        \max\{g_1 \tilde{\beta}_k - g_2\sgn(\tilde{\beta}_k) - g_0\} \le \max\{g_1 \tilde{\beta}_k - g_2(1 + \sgn(\tilde{\beta}_k))\} \;
    \end{displaymath}
    so $\tilde{\beta}_{k+1} \le g_1 \tilde{\beta}_k$. Thus, $\tilde{\beta}_{k} \not\in \RR_-$ implies that $\tilde{\beta}_{k+1} < \tilde{\beta}_{k}$, which shows that $\RR_-$ is also attractive and, hence, the asymptotic stability of the origin is global.
\end{proof}

The other asymptotic behavior is characterized next.

\begin{proposition} \label{prop:behav_2}
    Let $g_2 > g_0 > 0$ and $g_1 \in (0,1)$. The solutions of the difference inclusion~\eqref{eq:beta_error} satisfy the ultimate bound
    \begin{displaymath}
        \limsup_{k \to \infty} |\tilde{\beta}_k| \le \frac{g_0+g_2}{1+g_1} \;.
    \end{displaymath}
\end{proposition}

\begin{proof}
    Consider the Lyapunov function $V_k = |\tilde{\beta}_k|$. We will show that $\Delta V_k < 0$ whenever $|\tilde{\beta}_k| > \frac{g_0+g_2}{1+g_1}$. We will consider four cases.

    The first case is $\frac{g_0+g_2}{1+g_1} < \tilde{\beta}_k \le \frac{g_0+g_2}{g_1}$. We have $\tilde{\beta}_{k+1} = g_1 \tilde{\beta}_k - g_2 - g_0 \le 0$, so $\Delta V_{k} = -g_1 \tilde{\beta}_k + g_0 + g_2 - \tilde{\beta}_k < 0$.  
    The second case is $\frac{g_0+g_2}{g_1} < \tilde{\beta}_k$. We have $\tilde{\beta}_k > 0$ and $\Delta V_{k} = g_1 \tilde{\beta}_k - (g_0 + g_2) - \tilde{\beta}_k < 0$. Thus, $\frac{g_0+g_2}{1+g_1} < \tilde{\beta}_k$ implies that $\Delta V_k < 0$. The third case is $-\frac{g_2-g_0}{g_1} \le \tilde{\beta}_{k} < -\frac{g_2-g_0}{1+g_1}$. We have $\tilde{\beta}_{k+1} = g_1 \tilde{\beta}_k + g_2 - g_0 \ge 0$ and $\Delta V_{k} = g_1 \tilde{\beta}_{k} + g_2 - g_0 + \tilde{\beta}_{k} < 0$. The last case is $\tilde{\beta}_{k} < -\frac{g_2-g_0}{g_1}$. We have $\tilde{\beta}_{k} < 0$ and $\Delta V_{k} = -(g_1 \tilde{\beta}_{k} + g_2 - g_0) + \tilde{\beta}_{k} < 0$. %= (1-g_1)\tilde{\beta}_{k} - (g_2-g_0) < 0 \;. 
    Finally, note that $\tilde{\beta}_{k} < -\frac{g_0+g_2}{1+g_1}$ implies that $\Delta V_k < 0$.
\end{proof}

\begin{figure}
    \centering
    \subfigure[Full dynamics. Error~\eqref{eq:ultimate_A}, computed by simulation of the nonlinear system~\eqref{eq:nonlinear} in feedback loop with the control~\eqref{eq:u} and the adaptive unit~\eqref{eq:low_pass}.]{\includegraphics[width=\columnwidth]{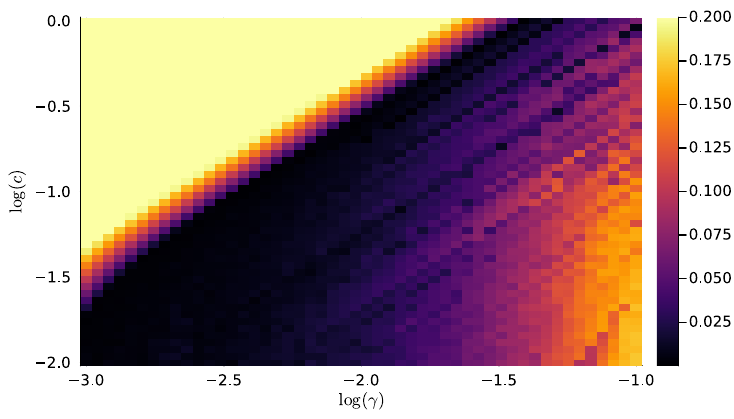}}
    \subfigure[Slow dynamics. Error~\eqref{eq:ultimate_A_beta}, computed using Propositions~\ref{prop:behav_1} and~\ref{prop:behav_2}.]{\includegraphics[width=\columnwidth]{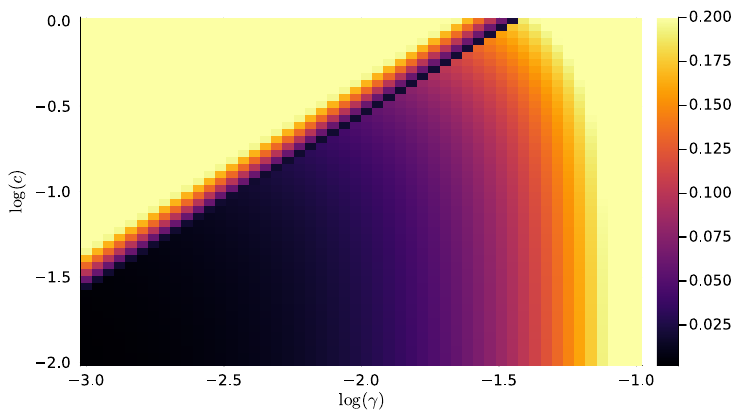}}
    \caption{Amplitude error. Comparison between the full and slow dynamics. Plant parameters $\lambda = 15$, $\xi = 0.1$, and $\omega_\rn = 8$, and adaptive parameters $\gamma \in [0.01, 1]$ and $c \in [0.01, 1]$ (in logarithmic scale).}
    \label{fig:heat_map}
\end{figure}

To illustrate the accuracy of the slow model, we simulate the full nonlinear system~\eqref{eq:nonlinear} with the control~\eqref{eq:u} and the adaptive unit~\eqref{eq:low_pass}. Using the simulated data, we compute the ultimate amplitude error
\begin{equation} \label{eq:ultimate_A}
    \limsup_{t \to \infty} |A(t) - A^\star|
\end{equation}
and compare it to the value predicted by~\eqref{eq:amp_of_beta},
\begin{equation} \label{eq:ultimate_A_beta}
    \limsup_{k \to \infty} |\hat{A}(\beta_k) - A^\star| \;.
\end{equation}
The ultimate bounds on $\beta_k$ in~\eqref{eq:ultimate_A_beta} are obtained using Propositions~\ref{prop:behav_1} and~\ref{prop:behav_2}. Figure~\ref{fig:heat_map} shows a comparison of~\eqref{eq:ultimate_A} and~\eqref{eq:ultimate_A_beta} for several values of $c$ and $\gamma$. The slow model reproduces the full model with reasonable accuracy, provided $\gamma$ is not too large.

We can see from Propositions~\ref{prop:behav_1} and~\ref{prop:behav_2} that~\eqref{eq:beta_error} undergoes a bifurcation at $g_2 = g_0$. By the definition of $g_0$ and $g_2$, we can see that the bifurcation occurs approximately at
\begin{displaymath}
    \gamma^\star = (e^{c\frac{\pi}{2\omega^\star}} - e^{-c\frac{\pi}{2\omega^\star}}) \beta^\star \;.
\end{displaymath}
% \textcolor{taya}{If we use $\omega = \omega_\rd$, that is true only for small $A^*$.}
This discrete-time bifurcation resembles that of a continuous-time Hopf bifurcation in that there is a transition from a stable equilibrium to an unstable equilibrium around which oscillations occur (see, e.g.,~\cite{guckenheimer}). 
% \begin{figure}
%     \centering
    % \includegraphics[width=\columnwidth]{beta_bifuc.pdf}
%     \caption{Behavior of $\beta(t_k)$ with $t_k \in \AC$ in the full model simulation ($A^* = 0.5$).}
%     \label{fig:beta_bifurc}
% \end{figure}
%\begin{figure}
%    \centering
%    \subfigure[Adaptation dynamics at control events $\beta(t_k), t_k \in \AC$.]{\includegraphics[width=\columnwidth]{beta_bifuc.pdf}}
%    \subfigure[Output $y(t)$ with $\gamma = \gamma^\star$.]{\includegraphics[width=\columnwidth]{theta_with_bifuc.pdf}}
%    \subfigure[Output $y(t)$ with $\gamma = 0.95\cdot\gamma^\star$.]{\includegraphics[width=\columnwidth]{theta_without_bifuc.pdf}}
%    \caption{Adaptation dynamics and output obtained via full model simulation with $A^\star = 0.5$, $\beta^\star = 0.0915$, and values of $\gamma$ close to its bifurcation value $\gamma^\star = 0.0075$.} %\textcolor{fernando}{Taya: we should explicitly indicate the equations that describe the model we are simulating}.}
%    \label{fig:theta_bifurc}
%\end{figure}

\begin{figure}
    \centering
    \includegraphics[width=\columnwidth]{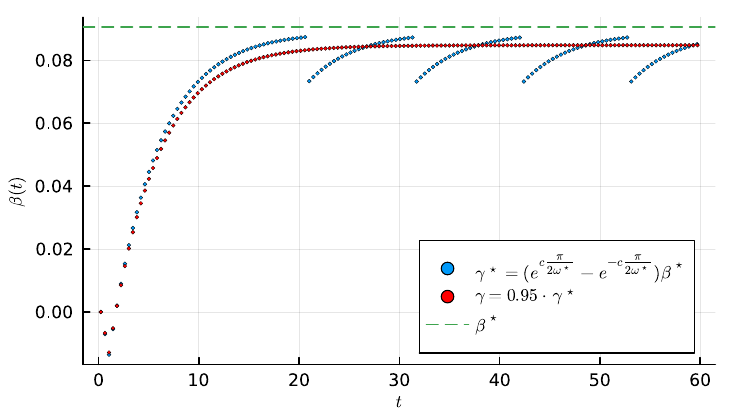}
    \caption{Adaptation dynamics obtained via full model simulation with $A^\star = 0.5$, $\beta^\star = 0.0915$, and values of $\gamma$ close to its bifurcation value $\gamma^\star = 0.0075$.} %\textcolor{fernando}{Taya: we should explicitly indicate the equations that describe the model we are simulating}.}
    \label{fig:theta_bifurc}
\end{figure}

%\begin{figure}[thpb]
%    \centering
%    \includegraphics[width=\columnwidth]{prop56_heatmap_analyt.pdf}
%    \caption{Error between the desired and obtained amplitudes, $|A(\beta_k) - A^*|$, for different controller parameter values $\gamma \in [0.01, 0]$ and $c \in [0.01, 1]$ (in logarithmic scale) calculated using the results of proposition~\ref{prop:behav_1}-\ref{prop:behav_2} and mapping~\eqref{eq:amp_of_beta}.}
%    \label{fig:frequencies}
%\end{figure}
%\begin{figure}[thpb]
%    \centering
%    \includegraphics[width=\columnwidth]{prop56_heatmap_simul.pdf}
%    \caption{Error between the desired and obtained in simulation of the closed loop system amplitudes for different controller parameter values $\gamma \in [0.01, 0]$ and $c \in [0.01, 1]$ (in logarithmic scale).}
%    \label{fig:frequencies}
%\end{figure}

\subsection{Parameter Optimization}

Our final objective is to choose $c$ and $\gamma$ such that the error $\tilde{\beta}_k$ remains as small as possible. According to Proposition~\ref{prop:behav_1}, we have the steady-state error
\begin{displaymath}
    J_\rs(\gamma,\beta^\star,c) =  \frac{g_0-g_2}{1-g_1} =
      \beta^\star - \gamma\frac{e^{-c\frac{\pi}{2\omega^\star}}}{1 - e^{-c\frac{\pi}{\omega^\star}}}
\end{displaymath}
when $g_2 \le g_0$, i.e., when $\tilde{\beta}_\fix$ is stable. According to Proposition~\ref{prop:behav_2}, the error is ultimately bounded by
\begin{displaymath}
    J_\ru(\gamma,\beta^\star,c) = \frac{g_0+g_2}{1+g_1} =
     \frac{(1 - e^{-c\frac{\pi}{\omega^\star}})\beta^\star + \gamma e^{-c\frac{\pi}{2\omega^\star}}}{1 + e^{-c\frac{\pi}{\omega^\star}}}
\end{displaymath}
when $g_2 > g_0$, i.e., when $\tilde{\beta}_\fix$ is unstable. For fixed $c > 0$, we propose the performance index
\begin{equation} \label{eq:J} 
    J(\gamma,\beta^\star,c) = 
    \begin{cases}
        J_\rs(\gamma,\beta^\star,c) & \text{if $\beta^\star \ge \gamma\frac{e^{-c\frac{\pi}{2\omega^\star}}}{1 - e^{-c\frac{\pi}{\omega^\star}}}$} \\
        J_\ru(\gamma,\beta^\star,c) & \text{if $\beta^\star < \gamma\frac{e^{-c\frac{\pi}{2\omega^\star}}}{1 - e^{-c\frac{\pi}{\omega^\star}}}$}
    \end{cases}
\end{equation}
and the robust optimization problem
\begin{equation} \label{eq:optim}
    \min_{\gamma > 0} \sup_{\beta^\star \in [\underline{\beta}, \overline{\beta}]} J(\gamma,\beta^\star,c) \;,
\end{equation}
where $\underline{\beta}$ and $\overline{\beta}$ are, respectively, lower and upper bounds on the (uncertain) burst width $\beta^\star$.

\begin{proposition} \label{prop:optimization}
    Consider the optimization problem~\eqref{eq:optim} with $J$ as in~\eqref{eq:J} and $g_i$, $i = 0,1,2$, as in~\eqref{eq:gs}. The solution is 
    \begin{displaymath}
        J(\gamma_\opt,\overline{\beta},c) = \min_{\gamma > 0} \sup_{\beta^\star \in [\underline{\beta}, \overline{\beta}]} J(\gamma,\beta^\star,c)
    \end{displaymath}
    with
    \begin{equation} \label{eq:gamma_opt}
        \gamma_{\opt} = \max\left\{\overline{\beta}\frac{1 - e^{-c\frac{2\pi}{\omega^\star}}}{3e^{-c\frac{\pi}{2\omega^\star}} - e^{-c\frac{3\pi}{2\omega^\star}}},\underline{\beta}\frac{1 - e^{-c\frac{\pi}{\omega^\star}}}{e^{-c\frac{\pi}{2\omega^\star}}}\right\} \;.
    \end{equation}
\end{proposition}

\begin{proof}
    Fist, we consider $\gamma$ and $c$ fixed, and solve the optimization problem $\sup_{\beta^\star \in [\underline{\beta}, \overline{\beta}]} J(\gamma,\beta^\star,c)$. We consider three exhaustive cases:
    \begin{enumerate}[i)]
        \item $\overline{\beta} < \gamma e^{-c\frac{\pi}{2\omega^\star}}/(1 - e^{-c\frac{\pi}{\omega^\star}})$. This implies that $\sup_{\beta^\star \in [\underline{\beta}, \overline{\beta}]} J(\gamma,\beta^\star,c) = J_\ru(\gamma,\overline{\beta},c)$.
        \item $\gamma e^{-c\frac{\pi}{2\omega^\star}}/(1 - e^{-c\frac{\pi}{\omega^\star}}) \in (\underline{\beta}, \overline{\beta}]$. This implies that
        \begin{displaymath}
            \sup_{\beta^\star \in [\underline{\beta}, \overline{\beta}]} J(\gamma,\beta^\star,c) = \max\left\{J_\ru\left(\gamma,\gamma\tfrac{e^{-c\frac{\pi}{2\omega^\star}}}{1 - e^{-c\frac{\pi}{\omega^\star}}},c\right), J_\rs(\gamma,\overline{\beta},c)\right\} \;,
        \end{displaymath}
%        where $J_\ru\left(\gamma,\frac{\gamma e^{-c\frac{\pi}{2\omega^\star}}}{1 - e^{-c\frac{\pi}{\omega^\star}}},c\right) = \gamma \frac{2e^{-c\frac{\pi}{2\omega^\star}}}{1 + e^{-c\frac{\pi}{\omega^\star}}}$.
        \item $\gamma e^{-c\frac{\pi}{2\omega^\star}}/(1 - e^{-c\frac{\pi}{\omega^\star}}) \le \underline{\beta}$. This implies that $\sup_{\beta^\star \in [\underline{\beta}, \overline{\beta}]} J(\gamma,\beta^\star,c) = J_\rs(\gamma,\overline{\beta},c)$.
    \end{enumerate}
    Now we minimize the cost over $\gamma$. Let
    \begin{displaymath}
        \underline{\gamma} = \underline{\beta}\frac{1 - e^{-c\frac{\pi}{\omega^\star}}}{e^{-c\frac{\pi}{2\omega^\star}}} \quad \text{and} \quad
        \overline{\gamma} = \overline{\beta}\frac{1 - e^{-c\frac{\pi}{\omega^\star}}}{e^{-c\frac{\pi}{2\omega^\star}}} \;.
    \end{displaymath}
    Since $J_\ru\left(\cdot,\overline{\beta},c\right)$ and $J_\rs\left(\cdot,\overline{\beta},c\right)$ are monotonically increasing and decreasing, respectively, we know that the minimum in~\eqref{eq:optim} corresponds to case ii),
    \begin{displaymath}
        \min_{\gamma > 0} \sup_{\beta^\star \in [\underline{\beta}, \overline{\beta}]} J(\gamma,\beta^\star,c) = \min_{\gamma \in [\underline{\gamma}, \overline{\gamma}]}\max\left\{J_\ru\left(\gamma,\gamma\tfrac{e^{-c\frac{\pi}{2\omega^\star}}}{1 - e^{-c\frac{\pi}{\omega^\star}}},c\right), J_\rs(\gamma,\overline{\beta},c)\right\} \;.
    \end{displaymath}
    The operands of the $\max$ operation coincide at $\gamma = \gamma^\star$ with
    \begin{displaymath}
        \gamma^\star \frac{2e^{-c\frac{\pi}{2\omega^\star}}}{1 + e^{-c\frac{\pi}{\omega^\star}}} = \overline{\beta} - \gamma^\star\frac{e^{-c\frac{\pi}{2\omega^\star}}}{1 - e^{-c\frac{\pi}{\omega^\star}}} \;,
    \end{displaymath}
    that is, with
    \begin{displaymath}
        \gamma^\star = \overline{\beta} \frac{1 - e^{-c\frac{2\pi}{\omega^\star}}}{3e^{-c\frac{\pi}{2\omega^\star}} - e^{-c\frac{3\pi}{2\omega^\star}}} \;.
    \end{displaymath}
    The minimum lies at $\gamma^\star$ whenever $\gamma^\star \in [\underline{\gamma},\overline{\gamma}]$.
    It is not difficult to verify that we always have $\gamma^\star \le \overline{\gamma}$, but we do not always have $\underline{\gamma} \le \gamma^\star$. The inequality $\underline{\gamma} > \gamma^\star$ presents itself when, e.g., $\underline{\beta}$ is close enough to $\overline{\beta}$, in which case the minimum is located at $\underline{\gamma}$. This is summarized in~\eqref{eq:gamma_opt}.
\end{proof}

Allow us to summarize by providing a step-by-step algorithm to tune the neuromorphic controller:
\begin{enumerate}
    \item Choose the desired amplitude value $A^\star$.
    \item Compute the value of the required pulse width, $\beta^\star$, using the solution of the HB equation with respect to $\omega^\star$ and $\beta^\star$ with fixed $A^\star$.
    \item Estimate the bounds of $\beta^\star$: $\beta^\star \in [\underline{\beta}, \overline{\beta}]$. The bounds $\underline{\beta}$ and $\overline{\beta}$ can be estimated for known bounded parameter uncertainties $\Delta\lambda, \Delta\xi$, and $\Delta\omega_\rn$. %Using the bounds on the uncertainties, compute bounds for the plant transfer function as $P(j\omega) \pm \Delta P(j\omega)$ and calculate $\underline{\beta}$ and $\overline{\beta}$ from the HB solution~\eqref{eq:HB} for $|P(j\omega) \pm \Delta P(j\omega)|$, respectively.
    \item Fix the time constant $c$ and compute the optimal gain~\eqref{eq:gamma_opt}.
%    \begin{equation*}
%        \gamma_{\opt} = \overline{\beta}\frac{1 - e^{-c\frac{2\pi}{\omega}}}{3e^{-c\frac{\pi}{2\omega}} - e^{-c\frac{3\pi}{2\omega}}}.
%    \end{equation*}
%    solving the optimization problem~\eqref{eq:optim}.
\end{enumerate}

\begin{figure}[thpb]
    \centering
%    \subfigure[Slow dynamics error $J(\gamma, \beta^\star, c)$.]{\includegraphics[width=\columnwidth]{surface_beta.pdf}}
    \includegraphics[width=\columnwidth]{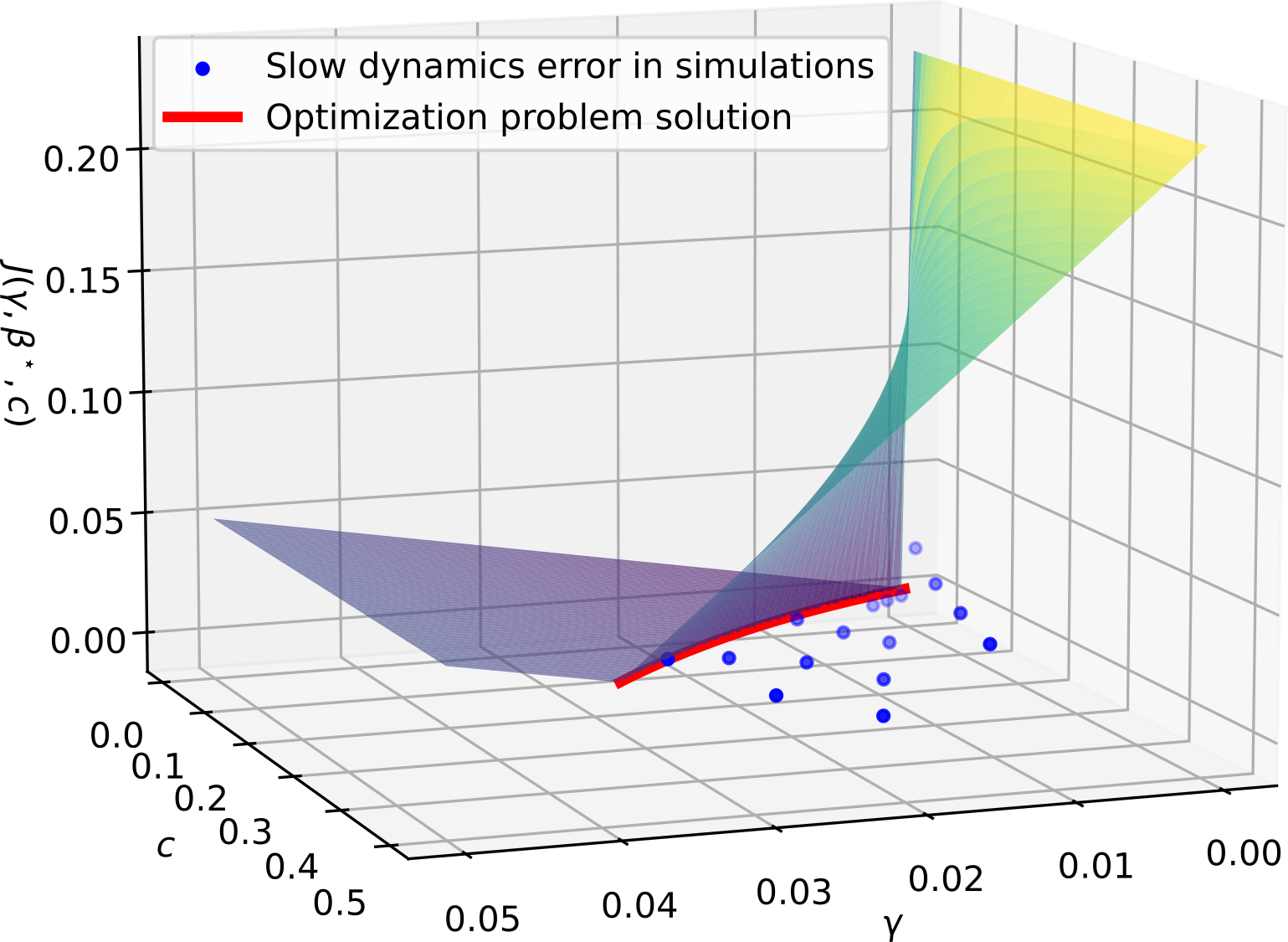}
%    \subfigure[Amplitude error $|A - A^\star|$.]{\includegraphics[width=\columnwidth]{surface_A.pdf}}
    \caption{Slow dynamics error $J(\gamma, \beta^\star, c)$. Surface generated by $\sup_{\beta^\star \in [\underline{\beta}, \overline{\beta}]} J(\gamma,\beta^\star,c)$, where $\underline{\beta} = 0.8\cdot\beta^\star = 0.0732$ and $\overline{\beta} = 2.5\cdot\beta^\star = 0.2288$, and burst width error (blue dots) obtained via simulations closed-loop system~\eqref{eq:u},~\eqref{eq:beta_causal}, and~\eqref{eq:nonlinear} for $\beta \in [\underline{\beta}, \overline{\beta}]$ and $\gamma_{\opt}$ computed from~\eqref{eq:gamma_opt}.} 
    %\textcolor{fernando}{Taya: I still think that having $\beta^\star$ on both sides of the inclusion is confusing. I think it's more clear to simply put the numerical values of $\underline{\beta}$ and $\overline{\beta}$. Also, we should explicitly indicate the equations that describe the model we are simulating}.}%amplitude error $|A - A^\star|$}
    \label{fig:beta_surface}
\end{figure}

% \textcolor{fernando}{Taya: We could use space better by doing something like this:}
% \begin{center}
%  \includegraphics[width=1.0\linewidth]{surface_simul_2.pdf}
% \end{center}
% \textcolor{fernando}{Also, we could add a plot of the optimal gain~\eqref{eq:gamma_opt} on the $\gamma-c$ plane.}

%%%%%%%%%%%%%%%%%%%%%%%%%%%%%%%%%%%%%%%%%%%%%%%%%%%%%%%%%%%%%%%%%%%%%%%%%%%%%%%%
\section{CONCLUSIONS AND FUTURE WORKS}

We introduced a general modeling architecture for neuromorphic control systems that aligns with the event-based nature of neuromorphic hardware. By adopting an input-output perspective, we formulated control problems that respect neuromorphic sensing, computing, and actuation constraints, while remaining amenable to classical control-theoretical tools. Applied to the neuromorphic pendulum, our framework enabled the derivation of rigorous performance bounds and revealed a discrete-time Hopf-like bifurcation in the closed-loop system. Interestingly, robust optimization of the control error led to an optimal solution that coincides with the bifurcation point, uncovering a meaningful trade-off between performance and robustness.

Future works include: reporting the experimental results that confirm Propositions~\ref{prop:behav_1},~\ref{prop:behav_2}, and~\ref{prop:optimization}, formulating the separation principle rigorously, and providing general guidelines for choosing $g_\ra$ and $\alpha_\ra$.

%These results emphasize the value of a principled, systems-theoretic approach to neuromorphic control, bridging neuromorphic engineering and control theory. Beyond providing a unifying formalism, our work opens avenues for extending the analysis to more general hybrid neuromorphic systems and for incorporating biologically inspired features such as adaptability and plasticity into the design of robust, event-based feedback controllers.

%\begin{itemize}
%    \item Report experimental results.
%    \item Report analysis of other systems (crawler).
%    \item Formulate separation principle rigorously (do we need three time-scales)?
%    \item Prove existence of limit cycles more rigorously, using a hybrid-systems approach.
%    \item How to choose $g_\ra$ and $\alpha_\ra$ in a systematic fashion?
%\end{itemize}

%%%%%%%%%%%%%%%%%%%%%%%%%%%%%%%%%%%%%%%%%%%%%%%%%%%%%%%%%%%%%%%%%%%%%%%%%%%%%%%%
\section{ACKNOWLEDGMENTS}

We would like to thank Romain Postoyan, Maurice Heemels, and Elena Petri for discussions on the complementary
paper~\cite{petri2025}. Both papers are submitted to the same invited session on neuromorphic systems and control. We also wish to acknowledge Prof. Leonid Fridman for the many helpful discussions.

%%%%%%%%%%%%%%%%%%%%%%%%%%%%%%%%%%%%%%%%%%%%%%%%%%%%%%%%%%%%%%%%%%%%%%%%%%%%%%%%
%\addtolength{\textheight}{-3cm}   % This command serves to balance the column lengths

\bibliographystyle{IEEEtran}
\bibliography{bibfile}

\end{document}